\documentclass[a4paper,reqno]{amsart}
\usepackage{amssymb}
\usepackage{latexsym}
\usepackage{amsmath}
\usepackage{euscript}
\usepackage{graphics,color}
\usepackage[all]{xy}
\usepackage[margin=3cm]{geometry}
\usepackage{MnSymbol} 

\def\cA{{\mathcal A}}   \def\cB{{\mathcal B}}   \def\cC{{\mathcal C}}
\def\cD{{\mathcal D}}

      \def\cU{{\mathcal U}}
\def\cV{{\mathcal V}}   \def\cW{{\mathcal W}}

\newcommand{\be}{\begin{equation}}
\newcommand{\ee}{\end{equation}}
\newcommand{\ba}{\begin{eqnarray}}
\newcommand{\ea}{\end{eqnarray}}
\newcommand{\baa}{\begin{eqnarray*}}
\newcommand{\eaa}{\end{eqnarray*}}
\newcommand{\bb}{\begin{thebibliography}}
\newcommand{\eb}{\end{thebibliography}}

\newcounter{my}
\newcommand{\he}%
   {\stepcounter{equation}\setcounter{my}%
   {\value{equation}}\setcounter{equation}0%
   }%
\newcommand{\she}%
   {\setcounter{equation}{\value{my}}%
    }%

\newtheorem{pr}{Proposition}

\newtheorem{theorem}{Theorem}[section]

\newtheorem{definition}[theorem]{Definition}
\theoremstyle{definition}

\newtheorem{remark}[theorem]{Remark}

\numberwithin{equation}{section}

\newcommand{\bra}[1]{\langle\,{#1}}
\newcommand{\ket}[1]{\mid{#1}\,\rangle}
\newcommand{\hg}[2]{\,\mbox{}_{#1}F_{ #2}\!}

\newcommand{\argu}[3]{\left(\begin{array}{c} #1\\#2\end{array} ; #3\right)}

\title{Meta Algebras and Biorthogonal Rational Functions:\\
 The Hahn Case}
\author{Satoshi Tsujimoto, Luc Vinet, Alexei Zhedanov}

\address{Graduate School of Informatics, Kyoto University,
Yoshida-Honmachi, Kyoto, Japan 606-8501}

\address{IVADO and Centre de recherches math\'ematiques, Universit\'e de Montr\'eal, P.O. Box 6128, Centre-ville Station, Montr\'eal (Qu\'ebec), H3C 3J7}

\address{Department of Mathematics, School of Information, Renmin University of China, Beijing 100872,CHINA}

\begin{document}

\begin{abstract}
    The finite families of biorthogonal rational functions and orthogonal polynomials of Hahn type are interpreted algebraically in a unified way by considering the three-generated meta Hahn algebra and its finite-dimensional representations. The functions of interest arise as overlaps between eigensolutions of generalized and ordinary eigenvalue problems on the representation space. The orthogonality relations and bispectral properties naturally follow from the framework.
\end{abstract}

\maketitle

\section{Introduction}

This paper initiates a program aimed at extending the Askey scheme to biorthogonal rational functions (BRFs) from a representation theoretic perspective. The approach will treat in a combined way both BRFs and orthogonal polynomials (OPs) with this article focusing on terminating  $_3F_2$ series i.e. functions of the Hahn type. The program hinges on the introduction of meta algebras whose name indicates that they subsume the algebras of Askey-Wilson type \cite{crampe2021askey} known to encode the bispectral properties of the Askey-Wilson polynomials and their limits and specializations. One feature of these meta algebras is that they admit as subalgebra a two-generated algebra modelling a non-commutative plane \cite{gaddis2015two}. 

By and large, the broad strategy to provide a unified algebraic interpretation of the finite polynomial families of the Askey scheme and their rational function companions is as follows. First, construct the finite-dimensional two-diagonal representation of all three generators of the meta algebra. Second, introduce various bases for these modules that are defined as solutions of Generalized Eigenvalue Problems (GEVP) and their transpose or of ordinary Eigenvalue Problems (EVP) and also of their transpose, set up on the two-diagonal representation space and solved using the known actions of the generators of the meta algebra. Third, construct overlaps between these different bases to be identified as the special functions of interest. As a consequence of the fact that eigenbases are used, the resulting special functions are by construction bispectral and this is straightforwardly spelled out. Moreover, orthogonality relations are found between overlaps involving correspondingly the eigenbases of transposed problems. 

This general program will be realized here for functions of the $_{3}F_2$ type. The intent is to develop this algebraic treatment for polynomials and rational functions of the $_4F_3$ type and for the basic ($q-$analogs) $_{3}\phi _2$ and $_{4} \phi _3$ types subsequently.

The key starting point is obviously the identification of the meta algebra in the present case, the meta Hahn one. Two previous papers have paved the way in this respect. In \cite{tsujimoto2021algebraic}, it was observed that three different operators $X, Y, Z$ act in a three-diagonal fashion on a monomial rational basis. It then proved possible to obtain the rational functions of Hahn type as solutions of the difference equation given by the GEVP involving $X$ and $Y$. The biorthogonal partner was provided by the adjoint problem. It was also seen that the GEVP defined rather in terms of $X$ and $Z$ provides a recurrence relation to show that the Hahn rational functions are biorthogonal and bispectral. It was further found that the difference operators $X, Y, Z$ formed an algebra that was called the rational Hahn algebra.

The idea of a unified algebraic treatment for both polynomials and rational functions of hypergeometric type was first put forward and realized in \cite{vinet2021unified} where the meta Hahn algebra $m\mathfrak{H}$ was introduced. Elements of the abstract representation theory of $m\mathfrak{H}$ were developed directly in eigenbases  of GEVP and EVP associated to the generators. It was shown that the overlaps between these bases are generally bound to be bispectral orthogonal polynomials or biorthogonal rational functions but it was only through the introduction of (differential or difference) models that the specific special functions were arrived at.

We now aim to provide a model-independent algebraic treatment of the Askey scheme enlarged to biorthogonal rational functions. It is hence appropriate to begin as a first step by revisiting the Hahn functions from this perspective and to record how the full characterization of these polynomial and rational functions can be synthetically derived from a remarkably simple algebra.

The organization of the paper will follow the general strategy sketched above. The framework for the joint algebraic interpretation of OPs and BRFs will be explained in more precise terms in Section 2. Working out the Hahn case will thereafter be the objective for the remainder of the article. The definition of the meta Hahn algebra $m \mathfrak{H}$ will be recalled in Section 3. How the Hahn algebra $\mathfrak{H}$ associated to the Hahn OPs embeds in $m\mathfrak{H}$ will also be spelled out. The finite-dimensional two-diagonal representation of $m \mathfrak{H}$ will be given in Section 4. The eigenvector bases will be explicitly constructed in Section 5. The representations of $m \mathfrak{H}$ on these bases will be discussed in Section 6. How the OPs and their properties emerge in this framework will be the object of Section 7 and the corresponding BRFs will be introduced and analyzed in Section 8. 
An outlook and perspectives will form Section 9.  Appendix A comprises a compendium of formulas for the actions of the generators in various bases.

\section{The General Framework} \label{gen}

A meta algebra (associated with finite families of functions) will have three generators $X, Z, V$ and will possess in particular a $(N+1)$-dimensional module $\mathcal{M}$ over $\mathbb{R}$,
equipped with a basis $\{\ket{n}, n = 0, \dots, N\}$ in which all generators act in a two-diagonal way and, a scalar product denoted by $\bra{x}\ket{y}$ for any two vectors $\ket{x},\ket{y} \in \mathcal{M}$. Corresponding to the embeddings of an Askey type algebra and of the associated rational algebra that they exhibit, meta algebras have the feature of encompassing a Leonard pair \cite{terwilliger2003introduction} formed by $V$ and $W$ (see below) and a GEVP- EVP analog involving $(X, Z)$ and $V$ \cite{vinet2021unified}.

\subsection{Bases for $\mathcal{M}$}

In addition to the basis $\{\,\ket{n}\,\}$, the following eigenbases of $\mathcal{M}$ associated to Generalized Eigenvalue Problems (GEVP) and ordinary Eigenvalue Problems (EVP) will hence be called upon.
\smallskip
\begin{itemize}
    \item GEVP bases $\{\,\ket{d_n}\,\}, \{\,\ket{d_n^*}\,\},\: n=0,\dots,N$:
    \begin{align}
        &(X-\lambda_nZ)\ket{d_n} =0,\\
        &(X^{\top}-\lambda_nZ^{\top})\ket{d_n^*} = 0.
    \end{align}
\end{itemize}
Here and in the following we shall not distinguished between the generators and their representation as operators acting on $\mathcal{M}$ since the context will always make clear what is understood. With this noted, $X^{\top}$ denotes the transpose of the operator $X$ and similarly for $Z$. The vectors involving an $*$ in their notation will always correspond to the transposed problems.
\begin{itemize}
    
    \item EVP bases $\{\,\ket{e_n}\,\},\{\,\ket{e_n^*}\,\}, \{\,\ket{f_n}\,\}, \{\,\ket{f_n^*}\,\},\: n=0,\dots,N$:
    \begin{align}
        &V\ket{e_n} = \mu _n \ket{e_n},\\
        &V^{\top}\ket{e_n^*} = \mu _n \ket{e_n^*},
    \end{align}
    and with $W=X+\rho Z$ where $\rho$ is a real parameter:
    \begin{align}
        &W\ket{f_n} = \nu _n \ket{f_n},\\
        &W^{\top}\ket{f_n^*} = \nu _n \ket{f_n^*}.
    \end{align}
\end{itemize}
We have the following orthogonality relations:
\begin{align} 
    & \bra{e_m^*}\ket{e_n}=  \kappa_n^{-1} \delta_{m,n}, \label{orth:e}\\ 
    & \bra{f_m^*}\ket{f_n}= \zeta_n^{-1}\delta_{m,n}, \label{orth:f} \\ 
    & \bra{d_m^*}\mid Z \ket{d_n} =(\bra{d_m^*}\mid Z^{\top}) \ket{d_n} = w_n^{-1}\delta_{m,n}, \qquad m,n = 0,1,\ldots,N, \label{orth:d}
\end{align}
where the constants $\kappa_n$, $\zeta_n$, $w_n$ define the norms of the basis elements. Note that over $\mathbb{R}$, $\bra{u}\ket{v} = \bra{v}\ket{u}$ for $\ket{u}, \ket{v} \in \mathcal{M}$.
While the first two relations (\ref{orth:e}), (\ref{orth:f}), are well known, the third (\ref{orth:d}) could be less familiar. It is proven as follows. We have
\begin{equation}
    \lambda_n \bra{d_m^*}\ket{Z \mid d_n} = \bra{d_m^*}\ket{X \mid d_n} = (\bra{d_m^*}\ket{X^{\top})\mid d_n} = \lambda_m(\bra{d_m^* \mid Z^{\top}}) \ket {d_n}
\end{equation}
which implies $(\lambda_n -\lambda_m)\bra{d_m^*}\ket{Z \mid d_n} = 0 $ from where (\ref{orth:d}) follows if $\lambda_n \neq \lambda_m.$
The corresponding completeness relations take the form
\begin{align}
    \sum_{n=0}^N  \kappa_n \ket{e_n} \bra{e_n^*}\mid = 1, \label{compl:e}\\
    \sum_{n=0}^N  \zeta_n \ket{f_n} \bra{f_n^*}\mid = 1, \label{compl:f}\\
    \sum_{n=0}^N  w_n Z \ket{d_n}\bra{d_n^*}\mid  = 1. \label{compl:d}
\end{align}

\subsection{Overlaps}
The following set of functions of the discrete variable $n$ and labelled by $m$ arising as expansion coefficients between bases are the central entities.
\begin{itemize}

    \item The EVP - EVP overlaps

    \begin{equation}
        S_m(n) = \bra{e_m}\ket{f^*_n},
    \end{equation}
    \begin{equation}
        \Tilde{S}_m(n) = \bra{e_m^*}\ket{f_n}.  
    \end{equation} 
 \end{itemize}   
The orthogonality relations obeyed by these functions:
\begin{align}
    &\sum_{n=0}^N \Tilde{S}_m(n) S_{m'}(n) \zeta_n = \kappa_m^{-1} \delta _{m,m'}, \label{orthS}\\
    &\sum_{m=0}^N \Tilde{S}_m(n) S_m(n') \kappa_m = \zeta_n^{-1} \delta _{n,n',} \label{dualorthS}
\end{align}
are readily seen to follow from \eqref{orth:e}, \eqref{orth:f}, \eqref{compl:e}, \eqref{compl:f}. Since the functions $S_m(n)$ and $\Tilde{S}_m(n)$ are both overlaps between eigenbases of the Leonard pair $(V, W)$ forming an algebra of the Askey-Wilson type whose representations are known to be unitarizable, we can expect $S_m(n)$ and $\Tilde{S}_m(n)$ to involve at their core the same orthogonal polynomials. (See Section 5 of \cite{vinet2021unified} for ampler explanations of this point.) This will be explicitly observed in the following.

\begin{itemize}
     
    \item The GEVP - EVP overlaps
        \begin{equation}
        U_m(n) = \bra{e_m}\ket{d^*_n},
    \end{equation}
    \begin{equation}
        \Tilde{U}_m(n) = \bra{e_m^*}\ket{Z  \mid  d_n}.  
    \end{equation} 
\end{itemize}
Stemming from \eqref{orth:d}, \eqref{orth:e}, \eqref{compl:d}, \eqref{compl:e}, the orthogonalities between these functions read:
\begin{align}
    &\sum_{n=0}^N \Tilde{U}_m(n) U_{m'}(n) w_n = \kappa_m^{-1} \delta _{m,m'}, \label{orthUUtilde}\\
    &\sum_{m=0}^N \Tilde{U}_m(n) U_m(n') \kappa_m = w_n^{-1} \delta _{n,n'} \label{dualorthUUtilde}. 
\end{align}
One defining property of a meta algebra is that it contains a GEVP analog of a Leonard pair whose features will be fully spelled out later but which is such that $X$ and $Z$ are tridiagonal in the eigenbasis of $V$. Integrating this fact with the above definitions, it follows that $U_m(n)$ satisfies a generalized eigenvalue equation defined in terms of two tridiagonal matrix and as shown in \cite{zhedanov1999biorthogonal} is thus formed of a rational function which will have its biorthogonal partner contained in $\Tilde{U}_m(n)$. Again, this will all be confirmed as we proceed.

\smallskip

From this point onward we shall focus on the Hahn case.

\section{Meta Hahn Algebra}
\begin{definition}
   The meta Hahn algebra $m\mathfrak{H}$ is generated by $X$, $Z$ and $V$ with the defining relations:
\begin{align}
    [Z,X] &= Z^2 +Z, \label{mHA1}\\
    [X,V] &= \{V,Z\} + V + 
    \xi I\label{mHA2}\\
    [V,Z] &= 2 X + 
    \eta I. \label{mHA3}
\end{align}  
\end{definition}
Consistently, $[A, B] = AB - BA$  and $\{ A , B \} = AB + BA.$  The Casimir element of $m\mathfrak{H}$ is given by:
\begin{equation}
Q=\{V, Z^2 +Z\}+2(X^2 +Z^2) +2\eta X+2(\xi +1)Z.   \label{Cas}
\end{equation}

The bispectral operators of the Hahn polynomials form a Leonard pair that realizes the Hahn algebra $\mathfrak{H}$ with generators $K_1$, $K_2$. Its defining relations are generically of the form:
\begin{align}
  [K_1, [ K_2, K_1]] &= a K_1^2 + b K_1 + c_1 K_2 + d_1 I,\label{hahnalg:1}\\
  [K_2, [K_1, K_2]] &= a \{K_1, K_2\} + b K_2 + c_2 K_1 + d_2, \label{hahnalg:2}
\end{align}
where $a$, $b$, $c_1$, $c_2$, $d_1$, $d_2$ are central parameters. (It is assumed that $a \ne 0$ in which case affine transformations of the generators bring the number of independent parameters to two.) A significant feature of the meta Hahn algebra is that the Hahn algebra embeds in it, i.e. $\mathfrak{H} \hookrightarrow m\mathfrak{H}$. Indeed by setting 
\begin{equation}
    K_1 = W = X + \rho Z, \qquad K_2 = V  \quad \text{with} \quad \rho \in \mathbb{R},
\end{equation}
and using the relations \eqref{mHA1}, \eqref{mHA2}, \eqref{mHA3} of $m\mathfrak{H}$ and the expression \eqref{Cas} for the Casimir element $Q$, one sees that $K_1$ and $K_2$ thus defined verify the Hahn relations \eqref{hahnalg:1} and \eqref{hahnalg:2} with
\begin{align}
    &a= 2, \;b= 2\rho-\xi+2\eta, \;c_1=-1, \;d_1 = -Q,\\
    &c_2=0, \;d_2=2\xi\rho.
\end{align}
\begin{remark}
    In \cite{vinet2021unified} the second and third defining relations of $m\mathfrak{H}$ are seen to involve an additional parameter and to be of the form:
    \begin{align}
        [X,V] &= \{V,Z\} + V + 
    \chi  X - \chi Z + \xi I, 
    \\
    [V,Z] &= 2 X + 
    \chi Z + 
    \eta I.    
    \end{align}
    It is readily seen that \eqref{mHA2} and \eqref{mHA3} are obtained from the latter by performing the automorphisms:
    \begin{align}
    V \to V+\frac{\chi(\chi+2)}{4}I,
    X \to X-\frac{\chi}{2} Z.
\end{align}
This corresponds to the freedom there is in splitting the Hahn algebra generator $K_1$ in two parts.
\end{remark}

\section{Two-diagonal Representation}
\begin{pr}\label{prop1}
    The two-diagonal representation of $m\mathfrak{H}$ on the finite-dimensional vector space $\mathcal{M}$ with basis $\{ \,\ket{n}, n = 0, \dots, N\}$ is given by the following actions of the generators $Z$, $X$ and $V$:
    \begin{align}
    Z\ket{n} &= -\ket{n} + a_{n} \ket{n+1},\label{ZonKetn}\\
    X\ket{n} &= (n-\alpha) \ket{n} - a_{n}(n-\beta) \ket{n+1}.\label{XonKetn}\\
    V\ket{n} &= (\beta-n)(n-\beta-1)\ket{n} - \dfrac{n(N+1-n)}{a_{n-1}}\ket{n-1},\label{VonKetn}
\end{align}
with $a_n, n=0,\dots, N$, normalization constants such that $a_N=0$, and $\alpha$, $\beta$ two parameters related to the algebra parameters $\xi, \eta$ 
and $N$ as follows:
\begin{align}
    \eta = -N+2\alpha, \quad \xi = (\beta+1)(N-\beta). \label{condpar}
\end{align}
\end{pr}
\begin{proof}
    
Assume that
\begin{align}
    Z\ket{n} &= c_n \ket{n} + a_{n} \ket{n+1},\\
    X\ket{n} &= d_n \ket{n} + b_{n} \ket{n+1}.
\end{align}
From (\ref{mHA1}), one obtains
\begin{align*}
    a_n a_{n+1} + a_{n} b_{n+1} - a_{n+1} b_n  =0,
\end{align*}
which leads to
\begin{align}
    b_n = a_n(-n + \beta).
\end{align}
One then finds that $c_n (c_n+1)=0$.
Taking $c_n = -1$, it follows that
\begin{align}
    d_n = n - \alpha.
\end{align}
The first relation (\ref{mHA1}) thus gives the two-diagonal actions \eqref{ZonKetn} and \eqref{XonKetn} of $Z$ and $X$.
Furthermore, the second and third relations (\ref{mHA2})-(\ref{mHA3}) yield \eqref{VonKetn} for $V \ket{n}$ with the conditions \eqref{condpar} required to ensure a finite-dimensional truncation.
\end{proof}
Although, this readily follows from Proposition \ref{prop1}, we shall record for convenience the action of the transposed generators:

\begin{align}
    Z^{\top}\ket{n} &= -\ket{n} + a_{n-1} \ket{n-1},\label{ZTonKetn}\\
    X^{\top}\ket{n} &= (n-\alpha) \ket{n} - a_{n-1}(n-1-\beta) \ket{n-1},\label{XTonKetn}\\
    V^{\top}\ket{n} &= (\beta-n)(n-\beta-1)\ket{n} - \dfrac{(n +1)(N-n)}{a_{n}}\ket{n+1}.\label{VTonKetn}
\end{align}

\section{(Generalized) Eigenbases} \label{Secion 5}

The various eigenbases of $\mathcal{M}$ defined in Section \ref{gen} can now be explicitly constructed given the above two-diagonal representation of $m\mathfrak{H}$. The results are gathered below.
\begin{pr}
The solutions to the GEVPs and EVPs of interest are given as follows in terms of expansions over the basis $\{\, \ket{n}, n = 0, \dots, N\}${\rm :}
  \begin{itemize}
  
\item {$X \ket{d_n} = \lambda_n Z \ket{d_n}$}
\begin{align}
    \lambda_n &=\alpha-n,\\
    \ket{d_n} &= \sum_{\ell=0}^{N} \dfrac{a_n a_{n+1}\cdots a_{N-1}}{a_{\ell} a_{\ell+1}\cdots a_{N-1}}\frac{(n-N)_{N-\ell}(n-N-\alpha+\beta+1)_{N-n}}{(n-N)_{N-n}(n-N-\alpha+\beta+1)_{N-\ell}} \ket{\ell}.
    \label{def:d_n}
\end{align}

\item {$X^{\top} \ket{d_n^*} = \lambda_n Z^{\top} \ket{d_n^*}$}
\begin{align}
    \lambda_n &=\alpha-n,\\
    \ket{d_n^*} &= \sum_{\ell=0}^{N} \dfrac{a_{0}a_1 \cdots a_{n-1}}{a_{0}a_1 \cdots a_{\ell-1}}    \frac{(-n)_{\ell}(-n+\alpha-\beta)_n}{(-n)_{n}(-n+\alpha-\beta)_{\ell}} \ket{\ell}.
    \label{def:d_n^*}
\end{align}

\item {$V \ket{e_n} = \mu_n \ket{e_n}$}
\begin{align}
    \mu_n &=(\beta-n)(n-\beta-1),\\
    \ket{e_n} &= \sum_{\ell=0}^{N} \dfrac{a_0 a_1 \cdots a_{\ell-1}}{a_0 a_1 \cdots a_{n-1}}
    \frac{n! (-N)_{n}(-n,n-2\beta-1)_{\ell}}{\ell! (-N)_{\ell} (-n,n-2\beta-1)_{n}} \ket{\ell}.
    \label{def:e_n}
\end{align}

\item {$V^{\top} \ket{e_n^*} = \mu_n e_n^*$}
\begin{align}
    \mu_n &=(\beta-n)(n-\beta-1),\\
   \ket{e_n^*} &= \sum_{\ell=0}^{N} \dfrac{a_{\ell}a_{\ell+1}\cdots a_{N-1}}{a_{n}a_{n+1}\cdots a_{N-1}}  \frac{(N-n)!(-N)_{N-n}(n-N,-N-n+2\beta+1)_{N-\ell}}{(N-\ell)! (-N)_{N-\ell}(n-N,-N-n+2\beta+1)_{N-n}} \ket{\ell}.
    \label{def:e_n^*}   
\end{align}

\item {$ W\ket{f_n} = (X+\mu Z) \ket{f_n} = \rho_n \ket{f_n}$} 
\begin{align}
    \rho_n &= n-\alpha-\mu, \\
    \ket{f_n} &= \sum_{\ell=0}^{N} \frac{a_{n}a_{n+1} \cdots a_{N-1}}{a_{\ell}a_{\ell+1} \cdots a_{N-1}}\frac{(n-N)_{N-\ell}(-N+\beta+\mu+1)_{N-n}}{(n-N)_{N-n}(-N+\beta+\mu+1)_{N-\ell}} \ket{\ell}. \label{def:f_n}
\end{align}

\item {$W^{\top}\ket{f_n^*} = (X^{\top}+\mu Z^{\top}) \ket{f_n^*} = \rho_n \ket{f_n^*}$} 
\begin{align}
    \rho_n &= n-\alpha-\mu, \\
    \ket{f_n^*} &= \sum_{\ell=0}^{N} \frac{a_{0}a_{1} \cdots a_{n-1}}{a_{0}a_{1} \cdots a_{\ell-1}}    \frac{(-n)_{\ell} (-\beta-\mu)_{n}}{(-n)_{n}(-\beta-\mu)_{\ell}} \ket{\ell}. \label{def:f_n^*}
\end{align}

\end{itemize}

Note importantly that the normalizations in \eqref{orth:e}, \eqref{orth:f}, \eqref{orth:d}, have been chosen so that $\kappa_n = 1$, $\zeta_n = 1$ and $w_n = -1$.
\end{pr}

\begin{proof}
Here is a sketch of the proof.   
First, the eigenvalues $\lambda_n$, $\mu _n$ and $\rho _n$ are straightforwardly identified from the diagonal part of the actions of the generators.
Second, consider to begin, the eigenvectors $\{\,\ket{e_n}\,\}_{n=0}^{N}$ of $V$ in the $N+1$-dimensional space $\mathcal{M}$.
From $V\ket{e_n} = \mu_n \ket{e_n}$, we have
\begin{align}
  \bra{\ell}\ket{(V-\mu_n)\mid e_n} =0, \quad (\ell = 0, 1, \ldots, N),
\end{align}
which upon using the action of $V^{\top}$ on $\bra{\ell}\mid$ amounts to
\begin{align}
  (\mu_\ell-\mu_n)\bra{\ell}\ket{e_n} - \dfrac{(\ell+1)(N-\ell)}{a_{\ell}}\bra{\ell+1}\ket{e_n} =0, \quad (\ell = 0, 1, \ldots, N),
  \label{ebase:cond}
\end{align}
where $\mu_{k} =(\beta-k)(k-\beta-1)$.
It is easy to solve (\ref{ebase:cond}) under the normalization condition $\bra{n}\ket{e_n}=1$ to find the explicit expression for $\ket{e_n}$ given in (\ref{def:e_n}).

Coming to the generalized eigenvectors $\{\,\ket{d_n}\,\}_{n=0}^{N}$ for $X\ket{d_n} = \lambda_n Z \ket{d_n}$, one has
\begin{align}
  \bra{N-\ell}\ket{(X-\lambda_n Z)\mid d_n } =0, \quad (\ell = 0, 1, \ldots, N),
\end{align}
which leads to   
\begin{align}
  -(\lambda_{N-\ell}-\lambda_{n})\bra{N-\ell}\ket{d_n} + \left((-N+\ell+1+\beta)-\lambda_n)\right)a_{N-\ell-1}\bra{N-\ell-1}\ket{d_n} =0,
  \label{dbase:cond}
\end{align}
where $\lambda_k = \alpha-k$,
for $\ell = 0,1,\ldots,N$.
Solving (\ref{dbase:cond}) under the normalization $\bra{n}\ket{d_n}=1$, yields
\begin{align}
    \bra{N-\ell}\ket{d_n} &= \prod_{j=0}^{\ell-1} \dfrac{(n+j-N)}{a_{N-j-1} (n+j -N-\alpha+\beta +1)}\bra{N}\ket{d_n} \\
    &=
    \dfrac{a_{n}\cdots a_{N-1}}{a_{N-\ell} \cdots a_{N-1}}\frac{(n-N)_{\ell}(n-N-\alpha+\beta+1)_{N-n}}{(n-N)_{N-n}(n-N-\alpha+\beta+1)_{\ell}}.
\end{align}
and hence the formula for $\ket{d_n} $ provided in (\ref{def:d_n}).

Similarly the eigenvectors  $\ket{e_n^*} $ (\ref{def:e_n^*}) and $\ket{d_n^* }$ (\ref{def:d_n^*})  of the problems involving the transposed operators
$V^{\top}$, $X^{\top}$ and $Z^{\top}$ are respectively obtained by solving
\begin{align}
    \bra{N-\ell}\ket{(V^{\top}-\mu_n) \mid e_n^*} =0, \quad
    \bra{\ell}\ket{(X^{\top}-\lambda_nZ^{\top}) \mid d_n^*} =0,
\end{align}
which entails
\begin{align}
    &
 (\mu_{N-\ell}-\mu_n)\bra{N-\ell}\ket{e_n^*} - \dfrac{(\ell+1)(N-\ell)}{a_{N-\ell-1}}\bra{N-\ell-1}\ket{e_n^*} =0,\\
   &  -(\lambda_{\ell}-\lambda_{n})\bra{\ell}\ket{d_n^*} + \left((-\ell+\beta)-\lambda_n)\right)a_{\ell}\bra{\ell+1}\ket{d_n^*} =0,
\end{align}
for $\ell =0,1,\ldots,N$.
One proceeds in the same fashion to obtain the expressions \eqref{def:f_n} and \eqref{def:f_n^*} for $\ket{f_n}$ and $\ket{f_n^*}$.

\end{proof}

Let us record here the following formula for the action of $Z$ on $\ket{d_n}$ :

\begin{align}
       Z \ket{d_n} &= \sum_{\ell=0}^{N} \dfrac{a_n a_{n+1}\cdots a_{N-1}}{a_{\ell} a_{\ell+1}\cdots a_{N-1}}\frac{(n-N)_{N-\ell}(n-N-\alpha+\beta+1)_{N-n}}{(n-N)_{N-n}(n-N-\alpha+\beta+1)_{N-\ell}} Z \ket{\ell} \nonumber\\
       &=\sum_{\ell=0}^{N} \dfrac{a_n a_{n+1}\cdots a_{N-1}}{a_{\ell} a_{\ell+1}\cdots a_{N-1}}\frac{(n-N)_{N-\ell}(n-N-\alpha+\beta+1)_{N-n}}{(n-N)_{N-n}(n-N-\alpha+\beta+1)_{N-\ell}} (-\ket{\ell} +a_{\ell} \ket{\ell+1})\nonumber\\
       &= - \sum_{\ell=0}^{N} \dfrac{a_n a_{n+1}\cdots a_{N-1}}{a_{\ell} a_{\ell+1}\cdots a_{N-1}}\frac{(n-N)_{N-\ell}(n-N-\alpha+\beta+2)_{N-n}}{(n-N)_{N-n}(n-N-\alpha+\beta+2)_{N-\ell}} \ket{\ell}\nonumber\\
       & = -\ket{d_n} \Bigr|_{\alpha\to \alpha-1} \label{Zdn}
\end{align}
where $a_{N} = 0$.

We further notice, that the eigenvectors are of the form:
\begin{align*}
    \ket{e_n} &=\, \ket{n}+\sum_{j=0}^{n-1} C_{n,j}^{(e)} \ket{j}, \quad 
    \ket{e_n^*}=\, \ket{n}+\sum_{j=n+1}^{N} C_{n,j}^{(e*)} \ket{j},\\
    \ket{f_n} &=\, \ket{n}+\sum_{j=n+1}^{N} C_{n,j}^{(f)} \ket{j}, \quad 
    \ket{f_n^*} =\, \ket{n}+\sum_{j=0}^{n-1} C_{n,j}^{(f*)} \ket{j},\\
    Z \ket{d_n} &= -\ket{n}+\sum_{j=n+1}^{N} C_{n,j}^{(Zd)} \ket{j}, \quad 
    \ket{d_n^*}=\, \ket{n}+\sum_{j=0}^{n-1} C_{n,j}^{(d*)} \ket{j}.
\end{align*}
in keeping with the normalization choices and that the orthogonality relations \eqref{orth:e}, \eqref{orth:f}, \eqref{orth:d}, with the given conventions, are thus explicitly seen to result.

\section{Representations of $m\mathfrak{H}$ on various bases}

Some bits of the representation theory of $m\mathfrak{H}$ on EVP and GEVP bases were developed directly \cite{vinet2021unified} without recourse obviously to the two-diagonal representation. It is here possible to proceed with much more ease and clarity since these basis eigenvectors are now explicitly known elements of a fully characterized representation space. In addition to providing in this Section (and in Appendix A) various matrix elements of the generators (in these various eigenbases that will be used in the treatment of the arising special functions), we shall also summarize some salient features of these representations. We shall generically denote by $O_{m,n}^{(b)}$ the $m,n$ entry of the matrix representing the operator $O$ in the basis $b \in \{d, d^*, e, e^*, f, f^*\}$.

\subsection{Representations in the $e$ and $e^*$ bases}
The $e$ and $e^*$ bases, respectively formed from the eigenvectors of $V$ and $V^{\top}$, are pivotal in that they enter in the construction of both pairs of functions $S, \Tilde{S}$ and $U, \Tilde{U}$. A key observation is that $X$ and $Z$ (and hence $W=X + \rho Z$) are tridiagonal in the basis $e$.

From the expression \eqref{def:e_n} for $\ket{e_n}$ and the actions of $Z, X$ on the standard basis, we obtain in addition to $V \ket{e_n} = \mu_n \ket{e_n}$,
\begin{align}
    Z \ket{e_n} &= Z^{(e)}_{n+1,n} \ket{e_{n+1}} +Z^{(e)}_{n,n} \ket{e_{n}} + Z^{(e)}_{n-1,n} \ket{e_{n-1}},
\label{action:Zone}\\
    X \ket{e_n} &=X^{(e)}_{n+1,n} \ket{e_{n+1}} + X^{(e)}_{n,n} \ket{e_{n}} + X^{(e)}_{n-1,n} \ket{e_{n-1}},
\end{align}
where
\begin{align}
&Z^{(e)}_{n+1,n}=a_n, \label{action:Zone:coe1}\\
&Z^{(e)}_{n,n}= -1-\dfrac{(n+1)(n-N)}{2(\beta-n)}+\dfrac{n(n-N-1)}{2(\beta-n+1)}, \label{action:Zone:coe2}  \\
&Z^{(e)}_{n-1,n}=\dfrac{n(n-N-1)(\beta+(-n+2)/2)(\beta+(-n+1-N)/2)}{4a_{n-1}(\beta -n+1/2)(\beta -n+1)^2(\beta -n+3/2)},\label{action:Zone:coe3}\\
&X^{(e)}_{n+1,n}=a_n (\beta-n),
\label{action:Xone:coe1}
\\ &X^{(e)}_{n,n}=\frac{N}{2}-\alpha,
\label{action:Xone:coe2}
\\
    &X^{(e)}_{n-1,n}= -(\beta-n+1)Z^{(e)}_{n-1,n}.
\label{action:Xone:coe3}
\end{align}
The tridiagonal actions of the transposed operators $X^{\top}$ and $Z^{\top}$ on $ \ket{e_n^*}$ is readily obtained from the formulas above.

\subsection{Representations in the $f$ and $f^*$ bases}
The $f$ and $f^*$ bases are made out of the eigenvectors of the linear pencil $W = X + \rho Z$ and of its transpose $W^{\top} = X^{\top} + \rho Z^{\top}$. The significant observation here is that $V$ is tridiagonal in the $f$ basis. Indeed a straigtforward computation gives
\begin{align}
    V \ket{f_n} &= V^{(f)}_{n+1,n} \ket{f_{n+1}} + V^{(f)}_{n,n} \ket{f_{n}} + V^{(f)}_{n-1,n} \ket{f_{n-1}}, \label{action:Vonf} 
\end{align}
with
\begin{align}
  &V^{(f)}_{n+1,n}=a_n(n - \beta - \mu)(n- N +\beta - \mu + 1), \label{vf+1}\\
&V^{(f)}_{n,n}=
(N - 2 n) \mu - 2 n\left( N - n\right) + \beta (N - \beta - 1), \label{vf}\\
&V^{(f)}_{n-1,n}=\dfrac{n(n - N - 1)}{a_{n-1}}.  \label{vf-1}
\end{align}

Naturally, $V^{\top}$ acts tridiagonally in the $f^*$ basis, see \eqref{action:VTonfs}, \eqref{coefvf*}. As alluded before, in conjunction with the observations made in the last subsection, this confirms that $V$ and $W$ form a Leonard pair with $V$ tridiagonal in the $f$ basis where $W$ is diagonal and conversely with $W$ tridiagonal in the $e$ basis that diagonalizes $V$. In this case, the set $\{\,\ket{n}\, \}$ plays the role of a split basis for the pair \cite{terwilliger2003introduction} and this yields in a simple way the finite-dimendional representations of the Hahn algebra $\mathfrak{H}$.

\subsection{Representations in the $d$ and $d^*$ bases}

We recall that these bases are obtained from the solutions of the GEVP $(X-\lambda_n Z)\ket{d_n} = 0$ and its adjoint. There are two observations to stress regarding the actions of the generators in these bases. 
The first one is that $V$ is represented by a upper Hessenberg matrix in the basis $d$, that is a matrix with zero entries above the first superdiagonal. See \eqref{Vond}. The transposed situation occurs for $V^{\top}$ in the basis $d^*$ \eqref{VTonds}.
The second observation is that $VZ$ and $V^{\top} Z^{\top}$ respectfully act tridiagonally on the basis $d$ and $d^*$. The reader is referred to Appendix A for the corresponding formulas and additional results on the actions of the operators in the various bases.
We thus see how the notion of Leonard pair generalizes to situations involving one GEVP and one EVP. The operators $X$ and $Z$ that define the GEVP are tridiagonal in the basis $e$ where V is diagonal while it is $VZ$ (or $VX$) that is tridiagonal in the basis $d$ generated by the GEVP solutions. Therein lies as we will see the bispectrality of the BRFs $U$ and $\Tilde{U}$.

\section{Hahn and Dual Hahn Polynomials}
How the Hahn (and dual Hahn) polynomials arise in the functions $S_m(n) = \bra{e_m}\ket{f_n^*}$ and $\Tilde{S}_m(n) = \bra{e_m^*}\ket{f_n}$ is shown next. The characterization of these polynomials \cite{koekoek2010hypergeometric} is of course classical. It is recorded in the present context to underscore that their orthogonality and bispectral properties are obtained quite easily from the meta Hahn algebra representations and more importantly that OPs and BRFs can be treated within this framework in a completely parallel fashion.

Recall \cite{koekoek2010hypergeometric} that the Hahn polynomials are defined as 
\begin{equation}
    Q_m (x; \hat{\alpha}, \hat{\beta}, N) = \hg{3}{2}\argu{-m, m+\hat{\alpha}+\hat{\beta}+1,-x}{-N,\hat{\alpha}+1}{1}, \quad m=0, \dots , N \label{HP}
\end{equation}
and the dual Hahn polynomials as
\begin{equation}
    R_m (\lambda(x), \hat{\alpha}, \hat{\beta}, N) = \hg{3}{2}\argu{-m, -x, x+\hat{\alpha}+\hat{\beta}+1}{-N, \hat{\alpha}+1}{1}, \quad \text{with} \quad \lambda(x) = x(x+\hat{\alpha}+\hat{\beta}+1). \label{dHP}
\end{equation}
Note that when $x$ is the discrete variable $n=0,\dots,N$, the dual Hahn polynomials are obtained from the Hahn ones by exchanging $m$ and $n$ in the latter:
$R_m (\lambda(n), \hat{\alpha}, \hat{\beta}, N)=  Q_n (m; \hat{\alpha}, \hat{\beta}, N)$. The notation
\begin{equation}
    (a_1, a_2, \dots, a_k)_n = (a_1)_n (a_2)_n \dots (a_k)_n.
\end{equation}
shall be used.

\subsection{Identification}
Here is the precise connection between the EVP-EVP overlaps and the Hahn polynomials.
\begin{pr} \label{propid}
The functions $S_m(n) = \bra{e_m} \ket{f_n^*}$ and $\Tilde{S}_m(n) = \bra{e_m^*}\ket{f_n}$ are both expressible in terms of Hahn polynomials as follows
\begin{align}
 S_m(n)&=
    \dfrac{a_{0}a_{1} \cdots a_{n-1}}{a_0 a_1 \cdots a_{m-1}}
   \dfrac{N! (-1)^n}{n! (N-m)!}\dfrac{(\hat{\alpha}+1)_{n}}{(m+\hat{\alpha} + \hat{\beta}+1)_{m}}
 Q_m (n; \hat{\alpha}, \hat{\beta}, N) \label{Shahn}\\
   \tilde S_m(n)&=
    \frac{a_{n}a_{n+1} \cdots a_{N-1}}{a_{m}a_{m+1}\cdots a_{N-1}} \dfrac{N!(-1)^n (\hat{\alpha}+1)_m }{m! (N-n)! (\hat{\beta}+1)_m} \dfrac{(\hat{\beta}+1)_{N-n}}{(2m+\hat{\alpha}+\hat{\beta}+2)_{N-m}} Q_m (n; \hat{\alpha}, \hat{\beta}, N), \label{sthahn}
\end{align}
with

\begin{equation}
    \hat{\alpha} = -1-\beta-\mu  \qquad \hat{\beta} = \mu-\beta-1.\label{newpar}
\end{equation}
\end{pr}

\begin{proof}
 Given the expressions \eqref{def:e_n} and \eqref{def:f_n^*} for $\ket{e_n}$  and $\ket{f_n^*}$ respectively, using $\bra{\ell}\ket{k} = \delta _{\ell,k}$ and simple transformations, one straightforwardly finds 
 \begin{equation}
     \bra{e_m}\ket{f_n^*} =   \dfrac{a_{0}a_{1} \cdots a_{n-1}}{a_0 a_1 \cdots a_{m-1}}
     \dfrac{N! (-1)^n}{n! (N-m)!}\dfrac{(-\beta-\mu)_{n}}{(-2\beta+m-1)_{m}}
    \sum_{\ell=0}^{N} \frac{(-n,-m,m-2\beta-1)_{\ell} }{\ell!(-N,-\beta-\mu)_{\ell}},
 \end{equation}
 which in view of definition \eqref{HP} is readily seen to yield \eqref{Shahn} under the identification \eqref{newpar}.
 The derivation of formula \eqref{sthahn} requires a little more effort. From \eqref{def:e_n^*} and \eqref{def:f_n}, it is immediate to find:
 \begin{align}
    \tilde S_m(n)=\bra{e_m^*}\ket{f_n} &=
    \frac{a_{n}a_{n+1} \cdots a_{N-1}}{a_{m}a_{m+1}\cdots a_{N-1}} \dfrac{(N-m)!(-N)_{N-m}(-N+\beta+\mu+1)_{N-n}}{(n-N)_{N-n}(m-N,-N-m+2\beta+1)_{N-m}} \nonumber\\
   &\times 
    \sum_{k=0}^{N}
     \frac{(n-N,m-N,-N-m+2\beta+1)_{k}}{k!(-N,-N+\beta+\mu+1)_{k}}   \label{ST1}
 \end{align}
 The hypergeometric sum in \eqref{ST1} needs to be transformed to identify the Hahn polynomial. The following two formulas will be used:
 \begin{equation}
     \hg{3}{2}\argu{a, b, c}{d, e}{1} = \frac{\Gamma (e) \Gamma (d+e-a-b-c)}{\Gamma (e-c) \Gamma (d+e-b-c) } \hg{3}{2}\argu{a, d-b, d-c}{d, d+e-b-c}{1}, \label{TF1}
  \end{equation}
  (corollary 3.3.5 in \cite{andrews1999special}) and 
  \begin{equation}
      \hg{3}{2}\argu{-n, b, c}{d,e}{1} = \frac{(e-c)_n}{(e)_n} \hg{3}{2}\argu{-n, d-b, c}{d, c-e-n+1}{1}, \label{TF2}
  \end{equation}
  (an exercise in \cite{bailey1935generalized} and a special case of Whipple's formula \cite{gasper2011basic}). From \eqref{TF1}, one has
  \begin{align}
    \hg{3}{2}\argu{-N-m+2\beta+1,n-N,m-N}{-N,-N+\beta+\mu+1}{1} = &\frac{\Gamma(\beta+\mu+1-N)\Gamma(\mu-\beta-n+N)}{\Gamma(\beta+\mu+1-n-m)\Gamma(\mu-\beta+m)} \nonumber\\
    &\times \hg{3}{2}\argu{-N-m+2\beta+1,-n,-m}{-N,\beta+\mu+1-n-m}{1} 
\end{align}
and \eqref{TF2} yields
\begin{align}
    \hg{3}{2}\argu{-N-m+2\beta+1,-n,-m}{-N,\beta+\mu+1-n-m}{1} 
    & =\dfrac{(-\beta-\mu)_n}{(m-\beta-\mu)_n} \hg{3}{2}\argu{m-2\beta-1,-n,-m}{-N,-\beta-\mu}{1} ,
\end{align}
where we have used $(a+1-n)_n = (-1)^n (-a)_n$. Combining these results and recalling that $\Gamma (a+n)/\Gamma (a) = (a)_n$, one finds
  \begin{align}
    \tilde S_m(n)&=
    \frac{a_{n}a_{n+1} \cdots a_{N-1}}{a_{m}a_{m+1}\cdots a_{N-1}} \dfrac{(N-m)!(-N)_{N-m}}{(n-N)_{N-n}(m-N)_{N-m}}\dfrac{(-N+\beta+\mu+1)_{N-n}}{(-N-m+2\beta+1)_{N-m}} \nonumber\\
   &\times 
   \dfrac{(\mu-\beta)_{N-n}}{(-N+\beta+\mu+1)_{N-n-m} (\mu-\beta)_m} 
   \dfrac{(-\beta -\mu)_n}{(m-\beta -\mu)_n} \hg{3}{2}\argu{m-2\beta-1, -n -m}{-N, -\beta - \mu}{1}.
\end{align}   
Numerous simplifications can now be performed. First one notes that
\begin{equation}
    \frac{(N-m)! (-N)_{N-m}}{(n-N)_{N-n} (m-N)_{N-m}} = (-1)^{N-n} \frac{N!}{(N-n)! m!}.
\end{equation}
From the simple identity $(-m-a+1)_m = (-1)^m(a)_m$, it follows that
\begin{equation}
    (-N-m+2\beta)_{N-m} = (-1)^{N-m} (2m -2\beta+1)_{N-m},
\end{equation}
and one can further show that
\begin{equation}
    \frac{(-N+\beta+\mu+1)_{N-n}}{(-N+\beta+\mu+1)_{N-n-m}} = (-1)^m (n-\beta-\mu)_m,
\end{equation}
and that
\begin{equation}
    \frac{(-\beta - \mu)_n}{(m-\beta-\mu)_n} = \frac{(-\beta-\mu)_m}{(n-\beta -\mu)_m}.
\end{equation}
Putting all this together, one arrives at

  \begin{align}
    \tilde S_m(n)&=
    \frac{a_{n}a_{n+1} \cdots a_{N-1}}{a_{m}a_{m+1}\cdots a_{N-1}} \dfrac{N!(-1)^n}{(N-n)! m!}\nonumber\\
   &\times 
   \dfrac{(\mu-\beta)_{N-n}(-\beta -\mu)_m}{(2m-2\beta)_{N-m} (\mu-\beta)_m} 
   \hg{3}{2}\argu{m-2\beta-1, -n -m}{-N, -\beta - \mu}{1},
\end{align}   
which is checked to coincide with \eqref{sthahn} upon substituting the parameters \eqref{newpar} and using the definition \eqref{HP} of the Hahn polynomials.
\end{proof}

\subsection{Orthogonality relations}

The fact that the functions $S_m(n)$ and $\Tilde{S}_m(n)$ are orthogonal by construction with respect to the variables $m$ and $n$ (recall \eqref{orthS}, \eqref{dualorthS}) can now be exploited to readily recover the orthogonality relations of the Hahn and dual Hahn polynomials and their normalizations. Remember that $\zeta_n = \kappa_n = 1, \forall n$ given the normalization of the eigenvectors that have been imposed in Section \ref{Secion 5}.

Substituting the expressions for $S_m(n)$ and $\Tilde{S}_m(n)$ of Proposition \ref{propid} in
\begin{equation}
    \sum_{n=0}^N \Tilde{S}_m(n)S_{m'}(n) = \delta_{m.{m'}},
\end{equation}
one finds
\begin{align}
    &\sum_{n=0}^N \frac{(\hat{\alpha}+1)_n (\hat{\beta}+1)_{N-n}}{n! (N-n)!} Q_m (n; \hat{\alpha}, \hat{\beta}, N)Q_{m'} (n; \hat{\alpha}, \hat{\beta}, N)\nonumber\\
    & \qquad = \frac{(-1)^m (m+\hat{\alpha} +\hat{\beta}+1)_{N+1} (\hat{\beta}+1)_m m!}{(2m+ \hat{\alpha}+ \hat{\beta}+1) (\hat{\alpha}+1)_m (-N)_n N!} \delta_{m,{m'}},
\end{align}
using
\begin{equation}
    (2m+\hat{\alpha}+\hat{\beta}+2)_{N-m} (m+\hat{\alpha}+\hat{\beta}+1)_m = \frac{(m+\hat{\alpha}+\hat{\beta}+1)_{N+1}}{(2m+\hat{\alpha}+\hat{\beta}+1)} \label{simplif}
\end{equation}
and $(-N)_n=(-1)^n N!/(N-n)!.$ This is seen to coincide exactly with the formula given in \cite{koekoek2010hypergeometric} if one recalls that
\begin{equation}
    \binom{\alpha+k}{k}= (-1)^k\frac{(-\alpha-k)_k}{k!}= \frac{(\alpha+1)_k}{k!}.
\end{equation}

The other orthogonality relation
\begin{equation}
     \sum_{m=0}^N \Tilde{S}_m(n)S_{m}({n'}) = \delta_{n.{n'}},
\end{equation}
gives (using again \eqref{simplif}) the dual relation
\begin{align}
    &\sum_{m=0}^N \frac{(2m+\hat{\alpha}+\hat{\beta}+1)(\hat{\alpha}+1)_m (-N)_m  N!}{(-1)^m(m+\hat{\alpha}+\hat{\beta}+1)_{N+1} (\hat{\beta}+1)_m m!} Q_m (n; \hat{\alpha}, \hat{\beta}, N)Q_m ({n'}; \hat{\alpha}, \hat{\beta}, N)\\
    &= \frac{(N-n)! n!}{(\hat{\alpha}+1)_n (\hat{\beta}+1)_{N-n}} \delta_{n,{n'}}
\end{align}
which is of course the orthogonality relation of the dual Hahn polynomials \cite{koekoek2010hypergeometric}.

\subsection{Bispectral properties}

The recurrence relation and difference equation of the Hahn polynomials follow in the present picture from the fact that they appear in overlaps between eigenvectors of two different EVPs.

\subsubsection{Recurrence relation}
Recall the EVPs defined in terms of the linear pencil $W=X+\mu Z$ and its transpose: $W\ket{f_n}=\rho_n\ket{f_n}$ and $W^{\top}\ket{f_n^*}=\rho _n\ket{f_n^*}$ with $\rho_n = n-\alpha-\mu$. 
From $(\bra{ f_n^*}\mid W^{\top}) \ket{e_m} = \bra{f_n^*}\ket{W\mid e_m}$, 
one has
\begin{align}
\rho_n \bra{f_n^*}\ket{e_m}&=
W_{m+1,m}^{(e)}
\bra{f_n^*}\ket{e_{m+1}} +
W_{m,m}^{(e)}
\bra{f_n^*}\ket{e_{m}} +
W_{m-1,m}^{(e)}
\bra{f_n^*}\ket{e_{m-1}}.    
\end{align}
The matrix elements $W_{i,j}^{(e)}$ of $W$ in the basis $e$ are readily obtained from those of $X$ and $Z$ given by formulas \eqref{action:Zone:coe1}-\eqref{action:Xone:coe3}.

This amounts to the  following recurrence relation for $S_m(n)$:
\begin{align}
    \rho_n S_m(n) =
W_{m+1,m}^{(e)}S_{m+1}(n)+
W_{m,m}^{(e)}S_m(n)+
W_{m-1,m}^{(e)}S_{m-1}(n)  . \label{recurS}
\end{align}
from where the one for the Hahn polynomials follow. Indeed, substituting the expression \eqref{Shahn} for $S_m(n)$ one finds,
\begin{align}
     (n-\alpha-\mu) Q_m(n) =
&W_{m+1,m}^{(e)} \frac{(N-m)(m+\hat{\alpha}+ \hat{\beta}+1)}{a_m (2m + \hat{\alpha}+ \hat{\beta}+1)(2m+\hat{\alpha}+ \hat{\beta}+2)}Q_{m+1}(n) \nonumber \\
&+W_{m,m}^{(e)}Q_m(n)\nonumber \\
&+W_{m-1,m}^{(e)}\frac{a_{m-1}(2m+\hat{\alpha}+ \hat{\beta}-1)(2m + \hat{\alpha}+ \hat{\beta})}{(N-m+1)(m+\hat{\alpha}+ \hat{\beta})}Q_{m-1}(n),   \label{recurQ} 
\end{align}
where we have suppressed the parameters of the Hahn polynomials $Q_m (x; \hat{\alpha}, \hat{\beta}, N)$. Replacing the original parameters by $\hat{\alpha}$ and $\hat{\beta}$ given in \eqref{newpar}, one finds for the matrix elements:
\begin{align}
  &W_{m+1,m}^{(e)}=-a_m(m+\hat{\alpha}+1),  \\
  &W_{m,m}^{(e)}+\alpha +\mu = \frac{N}{2} + \frac{(\hat{\beta}-\hat{\alpha})}{2(2n+\hat{\alpha}+ \hat{\beta}+2)} 
  \bigg[2n-N - \frac{2n(n-N-1)}{2n+\hat{\alpha}+ \hat{\beta}}\bigg],\\
  &\qquad\qquad =\frac{N}{2} + \frac{\hat{\beta}-\hat{\alpha}}{4} + \frac{(\hat{\alpha}^2-\hat{\beta}^2)(2N+\hat{\alpha}+ \hat{\beta}+2)}{4(2n+\hat{\alpha}+ \hat{\beta})(2n+\hat{\alpha}+ \hat{\beta}+2)},\nonumber\\
  &W_{m-1,m}^{(e)}=\frac{(m+\hat{\beta})m(m-N-1)(m+\hat{\alpha}+\hat{\beta})(m+N+\hat{\alpha}+\hat{\beta}+1)}{a_{m-1}(2m+\hat{\alpha}+\hat{\beta}+1)(2m+\hat{\alpha}+\hat{\beta})^2(2m+\hat{\alpha}+\hat{\beta}-1)}.
\end{align}
It is then seen that the recurrence relation \eqref{recurQ} can be presented in the standard form \cite{koekoek2010hypergeometric}
\begin{equation}
    nQ_m(n)=-A_m Q_{m+1}(n)+(A_m + C_m)Q_m(n)-C_mQ_{m-1}(n),
\end{equation}
where
\begin{align}
    &A_m=\frac{(m+\hat{\alpha}+\hat{\beta}+1)(m+\hat{\alpha}+1)(N-m)}{(2m+\hat{\alpha}+\hat{\beta}+1)(2m+\hat{\alpha}+\hat{\beta}+2)},\\
    &C_m=\frac{m(m+\hat{\alpha}+\hat{\beta}+N+1)(m+\hat{\beta})}{(2m+\hat{\alpha}+\hat{\beta})(2m+\hat{\alpha}+\hat{\beta}+1)}.
\end{align}

\subsubsection{Difference equation}
The difference equation of the Hahn polynomials is obtained in a similar fashion by using the EVP for $V$: $V\ket{e_n} = \mu _n \ket{e_n}$ with $\mu_n=(\beta-n)(n-\beta-1)$. From $\bra{f_n^*}\ket{V\mid e_m}=(\bra{ f_n^*\mid V^{\top}})\ket{e_m}$ we have
\begin{align}
    \mu_m S_m(n)
&= V^{\top(f*)}_{n+1,n} S_m(n+1) +V^{\top(f*)}_{n,n} S_m(n)+V^{\top(f*)}_{n-1,n} S_m(n-1)
\nonumber\\
&=V^{(f)}_{n,n+1} S_m(n+1) +V^{(f)}_{n,n} S_m(n)+V^{(f)}_{n,n-1} S_m(n-1).
\end{align}
Substituting the expression \eqref{Shahn} for $S_m(n)$, yields
\begin{align}
    \mu_m Q_m(n) =& 
- \frac{a_n(n+\hat{\alpha} + 1)}{(n+1)}V^{(f)}_{n,n+1} Q_m(n+1) \nonumber \\ 
&+V^{(f)}_{n,n} Q_m(n)- \frac{n}{a_{n-1} (n+\hat{\alpha})}V^{(f)}_{n,n-1} Q_m(n-1).\label{difQ}
\end{align}
Observe that 
\begin{equation}
\mu_n = -m(m-2\beta-1)-\beta(\beta +1).
\end{equation}
Converting to the parameters $\hat{\alpha}$ and $\hat{\beta}$ and recalling \eqref{vf-1}, \eqref{vf+1}, \eqref{vf} for the matrix elements, one finds
\begin{align}
 &V^{(f)}_{n,n+1} = \frac{1}{a_n} (n+1)(n-N), \\
 &V^{(f)}_{n,n-1} = a_{n-1} (n+\hat{\alpha})(n-N-\hat{\beta}-1),\\
 &V^{(f)}_{n,n}+ \beta(\beta +1)= 2n^2+n(\hat{\alpha}-\hat{\beta}-2N) - N(\hat{\alpha}+1).
\end{align}
Incorporating all that in \eqref{difQ} proves  indeed that the Hahn polynomials obey the familiar difference equation \cite{koekoek2010hypergeometric}:
\begin{equation}
    m(m+\hat{\alpha}+\hat{\beta}+1) Q_m(n) = B(n) Q_m(n+1) - \big[B(n) + D(n)\big] Q_m(n) + D(n) Q_m(n-1),
\end{equation}
where
\begin{align}
    & B(n) = (n+\hat{\alpha}+1)(n-N)\\
    & D(n) = n(n-\hat{\beta}-N-1).
\end{align}
Obviously, as for the recurrence relation, the same final results are obtained by initiating the computations with the functions $\Tilde{S}_m(n)$.

\section{Hahn Rational Functions}
This section will provide the promised algebraic interpretation of the following functions introduced in \cite{vinet2021unified}.
\begin{align}
&\cU_m(x;a,b,N)
=\frac{(-1)^m(-N)_m}{(b+1)_m}
\hg{3}{2}\argu{-x,-m,b+m-N}{-N,a-x}{1}, \label{cU}
\\
&\cV_m(x;a,b,N)
=\cU_m(N-x;b+2-a,b,N).\label{cV}
\end{align}
Their rational character is manifest and they will be seen to be biorthogonal partners. As they are defined in terms of $\hg{3}{2}$ series they are naturally said to be of Hahn type. Their connection with the meta Hahn algebra, jointly with the eponymous polynomials, further justifies this name. The normalization has been chosen so that
\begin{equation}
\lim_{x\to\infty} \cU_m(x;a,b,N) = 1.
\end{equation}
\subsection{Representation theoretic interpretation}
As already mentioned, BRFs arise from considering the GEVP-EVP overlaps $U_m(n)=\bra{e_m}\ket{d_n^*}$ and $\tilde U_m(n) = \bra{e_m^*}\ket{Z \mid d_n}$.
\begin{pr}
The functions $U_m(n)$ and $\tilde U_m(n)$ are respectively given as follows in terms of the rational Hahn functions $\cU_m(x;a,b,N)$ and $\cV_m(x;a,b,N)$:
\begin{align}
  U_m(n)&=  
\frac{a_0\dots a_{n-1}}{a_0\dots a_{m-1}}
\frac{(1-a)_n (1+b)_m}{n! (m+b-N)_m}\,\cU_m(n;a,b,N),  \label{Ucu}\\
\tilde U_m(n) &= 
-\frac{a_n\dots a_{N-1}}{a_m\dots a_{N-1}}\frac{(m+1)_{N-m}(a-b-1)_{N-n}}{(N-n)!(-N)_m (-b)_{N-2m}}\,\cV_m(n;a,b,N). \label{UTcv}
\end{align}
where
\begin{equation}
    a=\alpha-\beta, \qquad b-N=-2\beta-1. \label{parab}
\end{equation}
\end{pr}
\begin{proof}
 The identification of $\cU _m(n)$ in $U_m(n)$ is readily achieved by taking the scalar product of the vectors $\ket{e_m}$ and $\ket{d_n^*}$ respectively given by \eqref{def:e_n} and \eqref{def:d_n^*}. This yields
\begin{equation}
U_m(n)=  \dfrac{a_{0}a_1 \cdots a_{n-1}}{a_0 a_1 \cdots a_{m-1}}\dfrac{m! (-N)_{m} (-n+\alpha-\beta)_n}{(-n)_{n}(-m,m-2\beta-1)_{m}}   \sum_{\ell=0}^{N}
    \frac{(-n,-m,m-2\beta-1)_{\ell}}{\ell! (-N,-n+\alpha-\beta)_{\ell} }, 
\end{equation}
from where one gets formula \eqref{Ucu} using \eqref{cU} and \eqref{parab}.

Obtaining $\Tilde{U}_m(n)$ requires more algebraic transformations. From formulas \eqref{def:e_n^*} and \eqref{Zdn} for $\ket{e_n^*}$ and $Z\ket{d_n}$ one finds:
\begin{align}
 \tilde U_m(n) &= -
    \dfrac{a_{n}a_{n+1}\cdots a_{N-1}}{a_{m}a_{m+1}\cdots a_{N-1}}  
     \frac{(N-m)!(-N)_{N-m}(n-N-\alpha+\beta+2)_{N-n}}{(n-N)_{N-n}(m-N,-N-m+2\beta+1)_{N-m}}
 \nonumber\\
     &\quad \times \sum_{\ell=0}^{N} 
     \frac{(n-N,m-N,-N-m+2\beta+1)_{N-\ell}}{(N-\ell)!(-N,n-N-\alpha+\beta+2)_{N-\ell}}
\nonumber\\
&= -
  \dfrac{a_{n}a_{n+1}\cdots a_{N-1}}{a_{m}a_{m+1}\cdots a_{N-1}}  
     \frac{(N-m)!(-N)_{N-m}(n-N-\alpha+\beta+2)_{N-n}}{(n-N)_{N-n}(m-N,-N-m+2\beta+1)_{N-m}}
 \nonumber\\
     &\quad \times \hg{3}{2}\argu{n-N, m-N, -N-m+2\beta +1}{N, n-N-\alpha+\beta+2}{1}.
\end{align}
Now use \eqref{TF2} successively twice to arrive at:
\begin{align}
 \tilde U_m(n) &= -
  \dfrac{a_{n}a_{n+1}\cdots a_{N-1}}{a_{m}a_{m+1}\cdots a_{N-1}}  
     \frac{(-1)^{N-n}(N-m)!(-N)_{N-m}(-N+\alpha+\beta)_{N-n}}{(n-N)_{N-n}(m-N)_{N-m}(-N-m+2\beta+1)_{N-m}}
 \nonumber\\
     &\quad \times \hg{3}{2}\argu{n-N, -m, m-2\beta -1}{N, n-\alpha-\beta+1}{1}, \label{almost}
\end{align}
having taken note that
\begin{equation}
    (m+n-\alpha-\beta+1)_{N-n} = (-1)^{N-n}(-m-N+\alpha+\beta)_{N-n}.
\end{equation}
From definition \eqref{cV} one sees that 
\begin{equation}
    \hg{3}{2}\argu{n-N, -m, m-2\beta -1}{_N, n-\alpha-\beta+1}{1} = \frac{(-1)^m (N-2\beta)_m}{(-N)_m}\cV_m(n, \alpha-\beta, N-2\beta-1, N). \label{VV}
\end{equation}
Integrating this observation in \eqref{almost}, using the easily proven identities
\begin{equation}
    (-1)^m (N-\kappa+1)_m (\kappa-N)_{N-2m}= (-N-m+\kappa)_{N-m}
\end{equation}
and $(m+1)_{N-m}=(-1)^{N-m}(-N)_{N-m}$ as well as the fact that  $(-n)_n =(-1)^n n!$ yields the identification \eqref{UTcv} upon introducing $a$ and $b$ as per \eqref{parab}.

\end{proof}
\subsection{Biorthogonality}
A key result is the biorthogonality of the rational functions of Hahn type. 
\begin{pr} The rational functions $\cU_m(n;a,b,N)$ and $\cV_m(n;a,b,N)$ satisfy the following biorthogonality relations:
 \begin{align}
&\sum_{n=0}^N  
\cW(n)\,\cV_m(n;a,b,N)\,
\cU_{m'}(n;a,b,N) = h_m \delta_{m,m'},\label{1storth}\\
&\sum_{m=0}^N  
\cW^*(m)\,\cV_m(n;a, b,N)
\,\cU_m(n';a, b,N) = h_n^*\delta_{n,n'} \label{2ndorth}.   
\end{align}
where
\begin{align}
& h_m = \frac{(1,-N,m-N+b)_m(2m-N+b+1)_{N-2m}}{(b+1)_m(b-N+1)_N},
\\
& \cW(n) =
\frac{(a-b-1)_{N-n}(1-a)_n }{(-b)_{N}}\frac{N!}{n!(N-n)!},\\
&h_n^*=\frac{(1,2-a+b-N)_n}{(-N,1-a)_n},\\
& \cW^*(m) = 
\frac{(2-a+b-N)_N(b+1)_m}{(1,-N,m+b-N)_m(2m+b-N+1)_{N-2m}}.
\end{align}   
\end{pr}
\begin{proof}
 It was observed quite generally in \eqref{orthUUtilde} and \eqref{dualorthUUtilde} that the overlaps $U_m(n)$ and $\tilde U_m(n)$ are biorthogonal. Recall that we have chosen normalizations so that $\kappa_n =1$ and $w_n=-1$ $\forall n$. Hence substituting the expressions \eqref{Ucu} and \eqref{UTcv} for $U_m(n)$ and $\tilde U_m(n)$ in terms of $\cU_m(n;a,b,N)$ and $\cV_m(n;a,b,N)$ in

 \begin{align}
    &\sum_{n=0}^N (-1)\Tilde{U}_m(n) U_{m'}(n) = \delta _{m,m'}, \\
    &\sum_{m=0}^N (-1)\Tilde{U}_m(n) U_m(n')  =  \delta _{n,n'} ,
\end{align}
gives the formulas recorded in the above proposition. 
Note that the weight functions $\cW$ and $\cW^*$ are specified by requesting the normalizations $h_0=h_0^*=1$, respectively. Technically, for presentation convenience, in obtaining \eqref{1storth}, both sides have been multiplied by $N!/(-b)_N$ and the identities 
\begin{equation}
    N!=m!(m+1)_{N-m} \quad \text{and \quad}(-1)^N (2m+b-N+1)_{N-2m} = (-b)_{N-2m} \label{simpleid}
\end{equation}
are used. Similarly to get \eqref{2ndorth}, both sides have been multiplied by $(2-a+b-N)_N/N!$ and the identities $(-1)^n (-N)_n(N-n)!=N!$, \eqref{simpleid} and
\begin{equation}
    (2-a+b-N)_N=(-1)^{N-n}(2-a+b-N)_n(a-b-1)_{N-n}
\end{equation}
are called upon. 
\end{proof}

\subsection{Bispectrality of $\cU_m(n)$}

\subsubsection{Recurrence relation}
\begin{pr}
     The rational function $\cU_m(n; a, b, N)$ of Hahn type obeys the following recurrence relation
\begin{align}
   & (n-m-a)\cA_m (\cU_{m+1}(n;a,b,N)-\cU_{m}(n;a,b,N)) \nonumber\\
&+(n+m-a+b-N)\cC_m (\cU_{m-1}(n;a,b,N)-\cU_{m}(n;a,b,N))=a(2m+b-N) \,\cU_{m}(n;a,b,N)
\label{cU:RRalt}
\end{align}

where
\begin{align}
   &\cA_m=\frac{(m+b+1)(m+b-N)}{(2m+b-N+1)}, \label{cA}\\  
   &\cC_m=\frac{m(m-N-1)}{(2m+b-N-1)} \label{cC}. 
\end{align}
\end{pr}
\begin{proof}
    The recurrence relation for $U_m(n)=\bra{e_m}\ket{d_n^*}$ is obtained from 
\begin{align*}
    \lambda_n \bra{e_m}\ket{Z^{\top}\mid d_n^*}=
    \bra{e_m}\ket{\lambda_n Z^{\top}\mid d_n^*}=
    \bra{e_m}\ket{X^{\top}\mid d_n^*}
\end{align*}   
recalling that $\ket{d_n^*}$ are solutions of a GEVP and that $X$ and $Z$ act tridiagonally on the $e$-basis. More precisely, the identity $\bra{e_m}\ket{X^{\top}-\lambda_nZ^{\top} \mid d_n^*}=(\bra{e_m}\ket{X-\lambda _n Z)\mid d_n^*}= 0$ is readily seen to imply
    \begin{align}
&X^{(e)}_{m+1,m}U_{m+1}(n) +X^{(e)}_{m,m} U_{m}(n)+X^{(e)}_{m-1,m} U_{m-1}(n)
\nonumber\\
&=(\alpha-n)\left(Z^{(e)}_{m+1,m}U_{m+1}(n) +Z^{(e)}_{m,m} U_{m}(n)+Z^{(e)}_{m-1,m} U_{m-1}(n)\right) \label{recU}
\end{align}
from where the recurrence relation for $\cU_m(n; a, b, N)$ will be extracted. From \eqref{action:Zone:coe1}-\eqref{action:Xone:coe3}, it is observed that
\begin{equation}
  X^{(e)}_{m+1,m}=(\beta - m) Z^{(e)}_{m+1,m}, \quad  X^{(e)}_{m-1,m}= -(\beta -m +1)Z^{(e)}_{m-1,m}, \quad X^{(e)}_{m,m}=\frac{N}{2}- \alpha.
\end{equation}
As a result \eqref{recU} can be recast in the form:
\begin{align}
    (n-m-\alpha +\beta )Z^{(e)}_{m+1,m}U_{m+1}(n) &+ \left[\frac{N}{2}-\alpha +(n-\alpha)Z^{(e)}_{m,m}\right] U_{m}(n) \nonumber\\
    &+(n+ m-\alpha -\beta -1)Z^{(e)}_{m-1,m} U_{m-1}(n)=0.
\end{align}
  Inserting in this equation the expression \eqref{Ucu} of  $U_{m}(n)$ in terms of $\cU_m(n;a,b,N)$ and using \eqref{parab} one finds
  \begin{align}
     &(n-m-a)Z^{(e)}_{m+1,m}\frac{(m+b+1)(m+b-N)}{a_m(2m+b-N)(2m+b-N-1)}\cU_{m+1}(n) \nonumber\\ 
     &+ \left[-a+\frac{b}{2} +\frac{1}{2}+\left(n-a+\frac{b}{2}-\frac{N}{2}+\frac{1}{2}\right)Z^{(e)}_{m,m}\right] \cU_{m}(n) \nonumber\\
    &+(n+ m-a+b-N)Z^{(e)}_{m-1,m}\frac{a_{m+1}(2m+b-N-2)(2m+b-N-1)}{(m+b)(m+b-N-1)} \cU_{m-1}(n)=0.   
  \end{align}
Calling upon the expressions \eqref{action:Zone:coe1}, \eqref{action:Zone:coe2}, \eqref{action:Zone:coe3} for the matrix elements of $Z$ in the $e$-basis, multiplying by $(2m+b-N)$ and simplifying, one arrives at \eqref{cU:RRalt}.
\end{proof}
By rearranging \eqref{cU:RRalt}, it is possible to reexpress this recurrence relation in a nice GEVP form. 
\begin{pr}
    The rational function $\cU_m(n; a, b, N)$ of Hahn type verifies the following difference GEVP equation:
    \begin{align}
&(2m+b-N)\Big(\cA_{m}\,\cU_{m+1}(n;a,b,N)-(\cA_{m}-\cC_{m}-2a)\,\cU_{m}(n;a,b,N)-\cC_{m}\,\cU_{m-1}(n;a,b,N)\Big) \nonumber\\
&=(2n-2a+b-N)\Big(\cA_{m}\,\cU_{m+1}(n;a,b,N)-(\cA_{m}+\cC_{m})\,\cU_{m}(n;a,b,N)+\cC_{m}\,\cU_{m-1}(n;a,b,N)\Big),\label{cU:RR}
\end{align}
where $\cA_m$ and $\cC_m$ are again given by \eqref{cA} and \eqref{cC}.
\end{pr}
Note that the eigenvalues differs from those of the GEVP for the vectors $\ket{d_n}$ and $\ket{d_n^*}$. This is keeping with fact that in GEVPs linear combinations of the operators induce homographic transformations of the eigenvalues.

\subsubsection{Difference equation}
\begin{pr} \label{diifcU}
   The rational function $\cU_m(n; a, b, N)$ of Hahn type obeys the following difference equation  \begin{align}
 &   \cB_n \,\cU_{m}(n+1;a,b,N)
 - (\cB_n+\cD_{n}) \,\cU_{m}(n;a,b,N)
    + \cD_{n} \,\cU_{m}(n-1;a,b,N)\nonumber\\
& \qquad = m (m+b-N) \Big((a-n) \,\cU_{m}(n;a,b,N)+n\,\cU_{m}(n-1;a,b,N)\Big) \label{diffeqcU}
\end{align}
where
\begin{align}
  &\cB_n=(n-a)(n-a+1)(n-N),\label{diffeqcU_coecB}\\
   &\cD_{n}=(n-a)(n-a+b-N) n.
\end{align}
\end{pr}
\begin{proof}

The difference equation for $U_m(n)=\bra{e_m}\ket{d_n^*}$ is derived by focusing on the EVP $V\ket{e_m} = \mu _m \ket{e_m}$ involving the degree $m$ and the fact that $V^{\top}Z^{\top}$  acts tridiagonally on the GEVP basis vectors $\ket{d_n^*}$. From the relation
\begin{align*}
    \mu_m\bra{e_m}\ket{Z^{\top}\mid d_n^*}=
    (\bra{e_m\mid V} )\,Z^{\top}\ket{d_n^*}=
    \bra{e_m}\ket{V^{\top}Z^{\top}\mid d_n^*},
\end{align*} 
one finds
\begin{align} &(V^{\top}Z^{\top})_{n+1,n}^{(d*)} U_{m}(n+1)+(V^{\top}Z^{\top})_{n,n}^{(d*)}U_{m}(n)+(V^{\top}Z^{\top})_{n-1,n}^{(d*)}U_m(n-1)\nonumber\\
& =(\beta-m)(m-\beta-1) (-U_m(n)+a_{n-1} U_m(n-1))
\end{align}
with the help of \eqref{ZTond*} and recalling that $\mu_m =(\beta-m)(m-\beta-1)$. Inserting in this equation the expression \eqref{Ucu} for $U_m(n)$ in terms of $\cU_m(n)$, and writing for now $a=\alpha -\beta$ and $b=N- 2\beta -1$, one gets
\begin{align} &(V^{\top}Z^{\top})_{n+1,n}^{(d*)} a_n \frac{(n-\alpha-\beta +1)}{(n+1)} \cU_{m}(n+1)\nonumber \\
&+(V^{\top}Z^{\top})_{n,n}^{(d*)}\cU_{m}(n)+(V^{\top}Z^{\top})_{n-1,n}^{(d*)}\frac{n}{a_{n-1}(n-\alpha +\beta)}\cU_m(n-1)\nonumber\\
& =(\beta-m)(m-\beta-1) (-\cU_m(n)+\frac{n}{n-\alpha + \beta}\cU_m(n-1)).
\end{align}
Using at this point the expressions \eqref{VZtop1}, \eqref{VZtop2} and \eqref{VZtop3} for the matrix elements of $V^{\top}Z^{\top}$ in the basis $d^*$, after multiplying both sides of the equation by $(n-\alpha + \beta +1)$, one finds that $\cU_{m}(n;\alpha-\beta,N-2\beta-1,N)$ obeys the difference equation
\begin{align}
    & \hat{\cB}_n \,\cU_{m}(n+1;\alpha-\beta,N-2\beta-1,N)+ \hat{\cD}_n \,\cU_{m}(n-1;\alpha-\beta,N-2\beta-1,N)\nonumber\\
&\qquad + (\beta (\alpha-\beta)(\beta+1)-\hat{\cB}_n-\hat{\cD}_n) \,\cU_{m}(n;\alpha-\beta,N-2\beta-1,N)\nonumber\\
& = \mu_m\Big((n-\alpha+\beta)\,\cU_{m}(n;\alpha-\beta,N-2\beta-1,N) -n \,\cU_{m}(n-1;\alpha-\beta,N-2\beta-1,N)\Big) \label{nextlast}
\end{align}
where
\begin{align}
  &\hat{\cB}_n=(n-N)(n-\alpha+\beta)(n-\alpha+\beta+1),\\
   &\hat{\cD}_n=n (n-\alpha)(n-\alpha-1).
\end{align}

By rewriting the parameters $\alpha, \beta$ in terms of $a, b$ and rearranging \eqref{nextlast}, the difference equation in the variable $n$ satisfied by $\cU_m$ is found to be as stated in Proposition \ref{diifcU}.
\end{proof}

It is interesting to observe that the difference equation \eqref{diffeqcU} for $\cU_{m}(n;a,b,N)$ can be recast in the following form:
\begin{align}
     \cB_n \,\Big(\cU_{m}(n+1;a,b,N)-\cU_{m}(n;a,b,N)\Big)
& +\cD_{n,m} \,\Big(\cU_{m}(n-1;a,b,N)-\cU_{m}(n;a,b,N)\Big)
    \nonumber\\
&  = a m (m+b-N) \,\cU_{m}(n;a,b,N)
\end{align}
where $\cB_n$ given by \eqref{diffeqcU_coecB} and
\begin{align}
   &\cD_{n,m}=(n-m-a)(n+m-a+b-N)n.
\end{align}
Note that $\cD_{n,m}-\cD_{n} = -n m (m+b-N)$.

\subsection{Bispectrality of $\cV_m(n)$}

The rational functions $\cV_m(n;a,b,N)$ which are the biorthogonal partners of the functions $\cV_m(n;a,b,N)$ stand on their own. From the observation that they are obtained from $\cU_m(n;a,b,N)$ by a reflection of the variable and a substitution of parameters, they will necessarily be bispectral given that the  $\cU$ s are. It is nevertheless important to see that their bispectral properties independently follow from the present algebraic framework.

\subsubsection{Recurrence relation}
\begin{pr}
    The rational function  of Hahn type defined in \eqref{cV} satisfy the recurrence relation
    \begin{align}
   & (N-n-m-b+a-2)\cA_m (\cV_{m+1}(n;a,b,N)-\cV_{m}(n;a,b,N)) \nonumber\\
&+(-n+m+a-2)\cC_m \big(\cV_{m-1}(n;a,b,N)-\cV_{m}(n;a,b,N)\big)=(b-a+2)(2m+b-N) \,\cV_{m}(n;a,b,N)
\label{cV:RR}
\end{align}
with $\cA_m$ and $\cC_m$ given by \eqref{cA} and \eqref{cC}.
\end{pr}
\begin{proof}
    Let $\ket{\tilde d_n}= Z \ket{d_n}$. Recall that $\cV_m(n;a,b,N)$ is related to $\tilde U_m(n) = \bra{e_m^*}\ket{ Z\mid d_n}$ that is to $\bra{e_m^*}\ket{\tilde d_n}$.
    The first relation \eqref{mHA1} of the algebra shows that $\tilde d_n$ satisfies
\begin{align}
   (X+Z+1 - \lambda_n Z)\ket{\tilde d_n }= 0.
\end{align}
This is confirmed by
\begin{align}
    (X+Z+1-\lambda_n Z)\,Z \ket{d_n} = (XZ+Z^2+Z - Z X) \ket{d_n} =0.
\end{align}
Thus we obtain the 
relation
\begin{align*}
    \lambda_n \bra{e_m^*}\ket{ Z \mid \tilde d_n} &=
    \bra{e_m^*}\ket{\lambda_n Z^2 \mid d_n} =
    \bra{e_m^*}\ket{Z X \mid d_n} \nonumber\\
    &=\bra{e_m^*}\ket{XZ + Z^2+Z \mid d_n}  =
    \bra{e_m^*}\ket{X + Z+1 \mid \tilde d_n}
\end{align*}
that leads to the recurrence relations:
\begin{align}
  &(X_{m,m+1}^{(e)}+Z_{m,m+1}^{(e)}  )\tilde U_{m+1}(n) 
  +(X_{m,m}^{(e)}+Z_{m,m}^{(e)}+1)  \tilde U_{m}(n) +(X_{m,m-1}^{(e)} +Z_{m,m-1}^{(e)} ) \tilde U_{m-1}(n) 
\nonumber\\
&=\lambda_n\left(
Z_{m,m+1}^{(e)}  \tilde U_{m+1}(n)
+Z_{m,m}^{(e)}  \tilde U_{m}(n)
+Z_{m,m-1}^{(e)}  \tilde U_{m-1}(n)
\right).
\end{align}
As done before, rewriting in terms of $\cV_m(n;a,b,N)$ using \eqref{UTcv}, watching for the $m$-dependent factors and invoking \eqref{action:Zone:coe1}, \eqref{action:Zone:coe2}, \eqref{action:Zone:coe3}, \eqref{action:Xone:coe1}, \eqref{action:Xone:coe2} and \eqref{action:Xone:coe3} yields the
equation \eqref{cV:RR}.
\end{proof}
This equation coincides with the one obtained under the rearrangement  $a\to b-a+2, n \to N-n$ of the recurrence relation \eqref{cU:RR} satisfied by $\cU_m(n)$.

\subsubsection{Difference equation}
\begin{pr} \label{diifcV}
The rational function $\cV_m(n; a, b, N)$ of Hahn type defined in \eqref{cV} satisfies the difference equation  
\begin{align}
 &  \tilde\cB_{n,m} \Big(\cV_m(n+1;a,b,N) - \cV_m(n;a,b,N)\Big)
  + \tilde\cD_{n} \Big(\cV_m(n-1;a,b,N) - \cV_m(n;a,b,N)\Big)
  \nonumber \\
 & \qquad = m(m+b-N)(-a+b+2) \cV_m(n;a,b,N)
 \label{diffeqcV}
\end{align}
where
\begin{align}
  &\tilde\cB_{n,m} = (N-n-m+a-b-2)(-n+m+a-2)(N-n),\\
  &\tilde\cD_{n} = -(N-n+a-b-2)(N-n+a-b-1)\,n.
\end{align}
\end{pr}
\begin{proof}
Using the first and the second relations \eqref{mHA1}--\eqref{mHA2} of the algebra, 
one finds the following equations
\begin{align*}
    \mu_m \bra{e_m^*}\ket{Z \mid \tilde d_n} &=
    (\bra{e_m^*\mid V^{\top}})\, Z^2 \ket{d_n} =
    \bra{e_m^*}\ket{VZ^2 \mid d_n} \nonumber\\
    &=\bra{e_m^*}\ket{ (ZV+2X+\eta)Z \mid d_n}\nonumber\\
    &=\bra{e_m^*}\ket{ Z(VZ+\eta) + 2 XZ \mid d_n} \nonumber\\
    &=\bra{e_m^*}\ket{Z(VZ+\eta) + 2 (ZX-Z^2-Z)\mid d_n}\nonumber\\
    &=\bra{e_m^*}\ket{Z(VZ+\eta+2X-2Z-2) \mid d_n}
\end{align*}
and $Z \ket{\tilde d_n} =Z^2 \ket{d_n}= Z(-\ket{d_n} +a_n \ket{d_{n+1}})  = -\ket{\tilde d_n } + a_n\ket{\tilde d_{n+1}}$. Recalling in addition \eqref{Xond}.
these observations allow to introduce the difference equations:
\begin{align}
&\mu_m\left( - \tilde U_{m}(n)
+a_n \tilde U_{m}(n+1)
\right) \nonumber\\ 
&\quad =
((VZ)^{(d)}_{n+1,n}+2X^{(d)}_{n+1,n} -2 Z^{(d)}_{n+1,n} )\tilde U_{m}(n+1) \nonumber \\
&\qquad + ((VZ)^{(d)}_{n,n} +\eta + 2X^{(d)}_{n,n} -2 Z^{(d)}_{n,n} -2) \tilde U_{m}(n) + (VZ)^{(d)}_{n-1,n} \tilde U_{m}(n-1)
\end{align}
with $\eta=-N+2\alpha$. 
Substituting for $\Tilde{U}_m(n)$ the formula \eqref{UTcv}, paying attention to the remnants of the gauge factor and using \eqref{action:Zond}, \eqref{Xond}, \eqref{VZtop1}, \eqref{VZtop2}, \eqref{VZtop3}, \eqref{VZcoe}, 
gives the difference equation
\eqref{diffeqcV}.
\end{proof}
Here one can see that
\begin{align}
   \tilde\cB_{n,m} =\tilde\cB_{n,0} + m(m+b-N)(n-N).
\end{align}
Then \eqref{diffeqcV} can be rewritten in GEVP form as 
\begin{align}
 &  \tilde\cB_{n,0} \Big(\cV_m(n+1;a,b,N) - \cV_m(n;a,b,N)\Big)
  + \tilde\cD_{n} \Big(\cV_m(n-1;a,b,N) - \cV_m(n;a,b,N)\Big)
  \nonumber \\
 & \qquad = m(m+b-N) \Big((N-n) \cV_m(n+1;a,b,N) -(N-n+a-b-2)\cV_m(n;a,b,N) \Big).
\end{align}

\subsection{Contiguity relations}
\begin{pr}  
    The rational functions $\cU_m(n, a, b:N)$ satisfy the following contiguity relations
 \begin{align}
    &a \,\cU_m(n;a+1,b,N)  =(a-n) \,\cU_m(n;a,b,N)+n \,\cU_m(n-1;a,b,N), \label{cont1}
\end{align}
and
\begin{align}
 &\frac{a (2m+b-N)}{a-n}\cU_m(n;a+1,b,N) \nonumber\\
&\quad   = \cA_m \cU_{m+1}(n;a,b,N)-(\cA_m+\cC_m)\cU_m(n;a,b,N))+\cC_m \cU_{m-1}(n;a,b,N). \label{cont2}   
\end{align}
\end{pr}
\begin{proof}

One can show that
  \begin{align*}
    Z^{\top}\ket{d_n^*} = - \ket{d_n^*} \Big|_{\alpha\to \alpha+1}
\end{align*}
in the same way that $Z\ket{d_n} = - \ket{d_n} \Big|_{\alpha\to \alpha-1}$ was proven in \eqref{Zdn}.

On the one hand, noticing that the vector $\bra{e_m}\mid$ does not depend on $\alpha$, 
 one obtains
\begin{align}
 U_{m}(n)\Big|_{\alpha\to\alpha+1} 
 &=- \bra{e_m}\ket{Z^{\top}\mid d_n^*}, 
\\
 U_{m}(n)\Big|_{\alpha\to\alpha+k} &= \bra{e_m}\ket{(-Z^{\top})^k \mid d_n^*}.
\end{align}
From 
$ Z^{\top}\ket{d_n^*}=\, \ket{d_n^*} - a_{n-1}\ket{d_{n-1}^*},$
one thus arrives at the contiguity relation 
\begin{align}
   U_{m}(n)\Big|_{\alpha\to\alpha+1}=U_{m}(n)- a_{n-1} U_{m}(n-1) \label{CR:U}
\end{align}
which gives \eqref{cont1} when rewritten in terms of $\cU_m$.
On the other hand, applying the operator $Z$ to the vector $\bra{e_m}\mid$ is seen to give another contiguity relation, namely, 
\begin{align}
 &U_{m}(n)\Big|_{\alpha\to\alpha+1} =- (\bra{e_m \mid Z})\ket{d_n^*}\nonumber\\
&\quad = -Z^{(e)}_{m+1,m}U_{m+1}(n;\alpha)-Z^{(e)}_{m,m}U_{m}(n;\alpha)-Z^{(e)}_{m-1,m}U_{m-1}(n;\alpha), \label{CR2:U} 
\end{align}
which amounts to \eqref{cont2} in terms of $\cU_m$.
\end{proof}

\section{Conclusion}
This paper has begun the study of an enlargement of the Askey scheme to biorthogonal rational functions with the objective of offering a full characterization of these functions.  The approach is rooted in the introduction of so-called meta algebras that stand to subsume the algebras of Askey-Wilson type \cite{koornwinder2023charting}, \cite{koornwinder2018dualities}, \cite{mazzocco2016confluences}, \cite{granovskii1992mutual}, \cite{vinet2016hypergeometric} associated to the various families of orthogonal polynomials of the Askey scheme \cite{koekoek2010hypergeometric}. These last algebras embody the bispectrality of the hypergeometric polynomials of the Askey scheme and have numerous interesting features ans applications \cite{crampe2021askey}. The meta algebras are expected to similarly account for the bispectrality of the biorthogonal rational functions attached to the entries of the Askey scheme. In fact these meta algebras are furthermore poised to provide a unified interpretation of both the orthogonal polynomials and the biorthogonal rational functions of a given hypergeometric type since the corresponding (degenerate) Askey-Wilson algebra embeds in the meta algebra of relevance. The general approach involves constructing modules of the meta algebra and then considering overlaps between different eigenbases of this representation space. The rational functions appear when Generalized Eigenvalue problems (GEVP) are set up while polynomials arise when using rather the eigenvectors of the linear pencil of the operators appearing in the GEVP.

 By focusing on the terminating $\hg{3}{2}$  case, that is the polynomials and rational functions of Hahn type, this paper has launched and validated this program. The meta Hahn algebra $m\mathfrak{H}$ was introduced and from its representation theory the properties of the Hahn polynomials were recovered and the characterization of the rational functions of Hahn type and of their biorthogonal partners was obtained. The simplicity of $m\mathfrak{H}$ should be stressed. It is minimally quadratic and has for subalgebra a PBW deformation \cite{gaddis2013pbw} of the Artin-Schelter regular algebra $[X,Z]=Z^2$ known as the Jordan plane \cite{iyudu2014representation}. (See also \cite{benkart2013parametric}.) 

Specifically, it was shown that starting from the easily obtained two-diagonal (finite-dimensional) representations of $m\mathfrak{H}$, the full theory of both the Hahn polynomials and the Hahn rational functions unfolds. As a result, the biorthogonality, the recurrence relations and the difference equations of the rational functions of Hahn type were obtained and, as a nice bonus,
a synthetic and complete derivation of the Hahn algebra representations and of the fundamental properties of the Hahn and dual Hahn polynomials was provided. 

The paper sets the directions to pursue the program of characterizing the biorthogonal rational functions of $q$-Hahn  \cite{bussiere2022bispectrality}, Racah and $q$-Racah types. This will call for the introduction of the corresponding meta algebras and would essentially complete the finite rational extension of the Askey scheme. We intend to pursue this next. Subsequently the infinite dimensional version of this study should be elaborated to complete the univariate picture.

Further generalizations obviously come to mind and underscore the potential that the notion of meta algebra holds: moving beyond the Askey scheme or exploring multivariate situations and their meta algebras are some examples. Examining more deeply meta algebras and their potential occurrences in diverse areas also warrants attention.

\section*{Acknowledgments}

The authors have repeatedly been hosted by the home institutions of their colleagues and express their joint gratitude to these organizations. 
The research of ST is supported by JSPS KAKENHI (Grant Number 19H01792). LV is funded in part through a discovery grant of the Natural Sciences and Engineering Research Council (NSERC) of Canada. The work of AZ is supported by the National Science Foundation of China (Grant No.11771015).

\section*{Appendix: Actions of operators on several bases}

Note that $a_{-1}=a_{N}=0$.
\begin{align}
    &Z \ket{d_n} = -\ket{d_n} + a_n \ket{d_{n+1}},
\label{action:Zond}\\
    &X \ket{d_n} = \lambda_n Z \ket{d_n} = -\lambda_n \ket{d_n} + \lambda_n a_n \ket{d_{n+1}}, \label{Xond}\\
    &V \ket{d_n} = \sum_{j=\max(0,n-1)}^N V^{(d)}_{j,n} \ket{d_{j}}, \label{Vond}\\
    &Z^{\top} \ket{d_n^*} = -\ket{d_n^*} + a_{n-1} \ket{d_{n-1}^*}, \label{ZTond*}\\
    &X^{\top} \ket{d_n^*} = \lambda_n Z^{\top} \ket{ d_n^*}=- \lambda_n \ket{ d_n^* }+  \lambda_n a_{n-1} \ket{d_{n-1}^*}, \\
    &V^{\top} \ket{d_n^*} = \sum_{j=0}^{\min(n+1,N)} V^{\top(d*)}_{j,n} \ket{d_{j}^*}, \label{VTonds}\\
 &V^{\top}Z^{\top} \ket{d_n^*} = (V^{\top}Z^{\top})_{n+1,n}^{(d*)} \ket{d_{n+1}^*} + (V^{\top}Z^{\top})_{n,n}^{(d*)} \ket{d_{n}^*} + (V^{\top}Z^{\top})_{n-1,n}^{(d*)} \ket{d_{n-1}^*},\\
& VZ\ket{d_n} 
=(VZ)_{n+1,n}^{(d)} \ket{d_{n+1}}
+(VZ)_{n,n}^{(d)} \ket{d_{n}}
+(VZ)_{n-1,n}^{(d)} \ket{d_{n-1}},
\label{action:VZond}\\    
    &Z \ket{e_n}= Z^{(e)}_{n+1,n} \ket{e_{n+1}}+Z^{(e)}_{n,n} \ket{e_{n}} + Z^{(e)}_{n-1,n}\ket{e_{n-1}},\\
    &X \ket{e_n} =X^{(e)}_{n+1,n}\ket{e_{n+1}} + X^{(e)}_{n,n}\ket{e_{n}} + X^{(e)}_{n-1,n}\ket{e_{n-1}},\\
    &V \ket{e_n} = \mu_n \ket{e_n}, \\
&Z\ket{f_n} =\sum_{j=n}^N(-1)^{j+n+1}a_{n}a_{n+1}\cdots a_{j-1}\ket{f_{j}}, \\
&X\ket{f_n} =(n-\alpha) \ket{f_n}+ \sum_{j=n+1}^N(-1)^{j+n}\mu\, a_{n}a_{n+1}\cdots a_{j-1}\ket{f_{j}}, \\
 &V^{\top}\ket{f_n^*} = V^{\top(f*)}_{n+1,n} \ket{f^*_{n+1}} +V^{\top(f*)}_{n,n}\ket{f^*_{n}}+V^{\top(f*)}_{n-1,n}\ket{f^*_{n-1}},\label{action:VTonfs} 
\end{align}
where
\begin{align}
& V_{n-1,n}^{(d)} =\frac{n(n-N-1)}{a_{n-1}},\\
&V_{n,n}^{(d)} =(n-\alpha)(\alpha+\eta_1+1-n)+(\alpha-\beta)(\alpha+\beta+\eta_1+1-N),\\
&V_{j,n}^{(d)} =(\alpha-\beta)(\alpha+\beta+\eta_1+1-N)\, a_na_{n+1}\dots a_{j-1}\qquad (j>n),\\
&(V^{\top}Z^{\top})_{n+1,n}^{(d*)}=\frac{(n + 1)(N-n)}{a_n}, \label{VZtop1}\\
& (V^{\top}Z^{\top})_{n,n}^{(d*)}=(\alpha - n)(N - 2n)+ (\beta - n) \eta_1 + (\beta - N)(\beta + 1), \label{VZtop2}\\
& (V^{\top}Z^{\top})_{n-1,n}^{(d*)}=a_{n - 1}(\alpha - n)(n - \alpha - \eta_1 - 1), \label{VZtop3}\\
&
(VZ)^{(d)}_{n+1,n} = (V^{\top}Z^{\top})_{n,n+1}^{(d*)},\quad
(VZ)^{(d)}_{n,n} = (V^{\top}Z^{\top})_{n,n}^{(d*)},\quad
(VZ)^{(d)}_{n-1,n} = (V^{\top}Z^{\top})_{n,n-1}^{(d*)}, \label{VZcoe}\\
&V^{\top(f*)}_{n+1,n}=V_{n,n+1}^{(f)},\quad V^{\top(f*)}_{n,n}=V_{n,n}^{(f)},\quad 
V^{\top(f*)}_{n-1,n}=V_{n,n-1}^{(f)} \quad \text{see \eqref{vf+1}, \eqref{vf}, \eqref{vf-1}}.\label{coefvf*}
\end{align}

\bibliographystyle{unsrt} 
\bibliography{ref_meta1.bib} 
\end{document}